\def\year{2020}\relax
\documentclass[letterpaper]{article} 
\usepackage{aaai20}  
\usepackage{times}  
\usepackage{helvet} 
\usepackage{courier}  
\usepackage[hyphens]{url}  
\usepackage{graphicx} 
\usepackage{graphicx}  
\usepackage{tikz}
\usetikzlibrary{calc}
\usepackage{amsmath,amsthm,mathrsfs,bm}
\usepackage{amssymb}
\usepackage{multirow}

\newtheorem{theorem}{Theorem}

\def\r#1{{\eqref{#1}}}

\def\cl#1{{\cal #1}}
\def\math#1{$#1$}
\def\mand#1{$$#1$$}
\def\mld#1{\begin{equation}
#1
\end{equation}}

\definecolor{filllightgray}{rgb}{0.6,0.6,0.6}

\newcommand{\myfbox}[3][white]{%
\begin{center}%
\begin{tikzpicture}
\node[draw,inner sep=5pt,rounded corners=5pt,line width=1pt,fill=#1] at (0,0)
{\parbox{#2}{\centering #3}};
\end{tikzpicture}%
\end{center}%
}

\def\degin{s^{\text{in}}}
\def\degout{s^{\text{out}}}

\title{True Nonlinear Dynamics from Incomplete Networks}
\author{Chunheng Jiang, Jianxi Gao, Malik Magdon-Ismail\\
Rensselaer Polytechnic Institute\\ 
110 8th Street, Troy, NY 12180\\
\{jiangc4,gaoj8\}@rpi.edu, magdon@cs.rpi.edu 
}
\begin{document}

\maketitle

\begin{abstract}
We study \emph{nonlinear} dynamics on complex networks. Each vertex $i$ has a state $x_i$ which evolves according to a networked dynamics to a steady-state $x_i^*$. We develop fundamental tools
to learn the true steady-state of a small part of the network, \emph{without knowing the
full network}. A naive approach and the current state-of-the-art is to follow the dynamics of the observed partial network to local equilibrium. This dramatically fails to extract the true steady state. We use a mean-field approach to map the dynamics of the unseen part of the network to a single node, which allows us to recover accurate estimates of steady-state on as few as 5 observed vertices in domains ranging from ecology to social networks to gene regulation. Incomplete networks are the norm in practice, and we offer new ways to think about nonlinear dynamics when only sparse information is available.
\end{abstract}

\section{Dynamical Systems on Incomplete Networks}

The fundamental task in learning is to infer unknown quantities of interest
from incomplete data. We study learning complex nonlinear dynamics on networks from
incomplete data. Such problems are fundamental because 
complex nonlinear dynamical systems are ubiquitous, often
modeled as coupled
nonlinear ordinary differential equations (ODEs), for example
epidemic spreading \cite{pastor2001epidemic}, Michaelis-Menten
gene regulatory dynamics \cite{alon2006introduction,gao2016universal},
Lotka-Volterra ecological dynamics \cite{lotka1910contribution}.
A graph  \math{G=(V,E)} with \math{n\times n} (weighted)
adjacency matrix \math{A} is the backbone on which the
dynamical equations are coupled together.
We consider a general dynamics in which
each vertex \math{i} of \math{G} has a state \math{x_i}
which evolves according to a self-driving force and a
sum of interaction forces over neighbors
  \begin{equation}
    \dot{x}_i=f(x_i)+\sum_{j\in V}A_{ij}g(x_i,x_j).\label{eq:general-ODE}
  \end{equation}
 The functions \math{f(\cdot)} and \math{g(\cdot,\cdot)} are general and
 usually nonlinear,
 and the positive weighted connectivity matrix \math{A} modulates
  the interactions between vertices.
    Several instances of such dynamics with appropriate
  choices of  \math{f(\cdot)} and \math{g(\cdot,\cdot)}
  are shown in Table~\ref{tab:real-dynamics}. From an initial state,
  one can step forward in time, simulating the dynamics
  in \r{eq:general-ODE} until convergence to equilibrium states
  \math{x_i^*}.

The complete information setting in \r{eq:general-ODE} is unrealistic, and we must accept that in practice, only part of a network can be measured. Hence, we assume that a subgraph with \math{m} nodes \math{G^{(s)}=(V^{(s)},E^{(s)})} is known, with corresponding \math{m\times m} adjacency matrix
  \math{A^{(s)}}, where \math{V^{(s)}\subseteq V} and
  \math{E^{(s)}\subseteq E}
  (\math{s} for sampled). 
This paradigm is unavoidable when the full network is unmeasurable (for example, protein-protein interactions, metabolic and terrorist networks~\cite{stumpf2005sampling}). The paradigm is also useful when the full network is too large to handle, where one can deliberately sample a much smaller subnetwork for the analysis. Analyzing a full network from a sampled
subnetwork has been studied in several contexts,  e.g. to estimate average or total degree~\cite{kurant2011towards};
degree distributions and clustering coefficients~\cite{stumpf2005sampling,gjoka20132,seshadhri2014wedge}; shortest paths~\cite{leskovec2006sampling}; motif counts~\cite{klusowski2018counting}; vertex and edge counts~\cite{katzir2011estimating,kurant2012graph}.

\begin{table*}{\small\tabcolsep5pt\renewcommand{\arraystretch}{1.35}
\begin{tabular}{p{0.2\textwidth}p{0.05\textwidth}p{0.15\textwidth}p{0.45\textwidth}}
Applications & Vertex & State at vertex \math{i} & Dynamics \\
\hline
Ecological (\scalebox{0.85}{\citeyear{allee1949principles,hui2006carrying}}) & Species & Abundance & \math{\dot{x}_i=B_i+x_i(1-\frac{x_i}{K_i})(\frac{x_i}{C_i}-1)+\sum_jA_{ij}\frac{x_ix_j}{D_i+E_ix_i+H_jx_j}}\\
Regulatory (\scalebox{0.85}{\citeyear{alon2006introduction,karlebach2008modelling}}) & Gene & Expression level & \math{\dot{x}_i=-Bx_i^f+\sum_j A_{ij}R\frac{x_j^h}{x_j^h+1}}\\
Epidemic (\scalebox{0.85}{\citeyear{pastor2001epidemic,hufnagel2004forecast,dodds2005generalized}}) & Person & Infection rate & \math{\dot{x}_i=-Bx_i+\sum_j A_{ij}R(1-x_i)x_j}\\
\hline
\end{tabular}}
\caption{Examples of real systems with nonlinear interaction dynamics.}
\label{tab:real-dynamics}
\end{table*}

\begin{figure*}[ht]
\centering
\includegraphics[width=\linewidth]{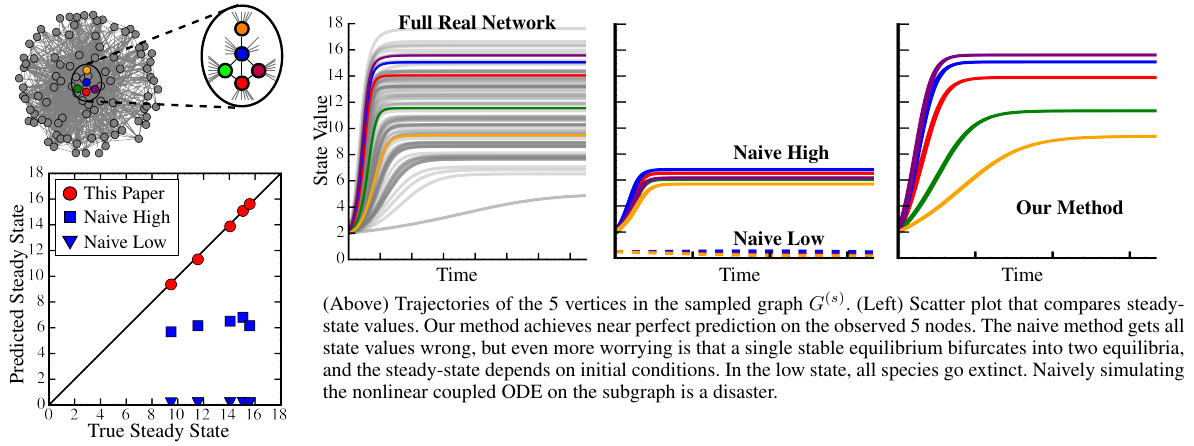}
    \caption{Predicting steady-state abundances of 5 species
      interacting in a larger 97 species ecological network. Predictions
      use only the interactions of those five species (incomplete information).
      \label{fig:example-5}}
  \end{figure*}
 Our task is to estimate true steady-states \math{x^*}
 for vertices in \math{V^{(s)}}, despite only seeing an incomplete
 network \math{G^{(s)}}.
 The state-of-the-art
 naive approach is to simulate \r{eq:general-ODE} on the subgraph
 \math{G^{(s)}}. For example, one may collect a social network from 
 Boston and run the epidemic model (Table \ref{tab:real-dynamics}) 
 to obtain the probability of each person to be infected. 
 The results are a dramatic and universal disaster,
 because the sub-social network of Boston is just a small
 part of a vast social network, and the people in Boston interact with people outside.
 That the naive method is bad is not surprising because the essence of the
 dynamics in \r{eq:general-ODE} are the interactions, and the subgraph
 is missing many of those interactions.
 Hence, not observing a large part of the network appears to be an
 insurmountable hurdle to learning the true steady states on the observed (small)
 part of the network.

 We develop a methodology to accurately predict \emph{true} steady states
 using only information local to \math{G^{(s)}}. We
 demonstrate the power of our approach in Figure~\ref{fig:example-5},
 for an ecological network where vertices are species and states are species-abundance.
 This ecological network of 97 species
 follows the symbiotic dynamics in Table~\ref{tab:real-dynamics},
 see~\cite{gao2016universal}.
 Let us see what happens when a biologist who is interested in
 \emph{five} species collects the relevant 5-vertex subgraph \math{G^{(s)}}
 and performs a naive simulation on the subgraph to get a steady-state.
 The results in Figure~\ref{fig:example-5} are as expected. The
 naive simulation is wrong. Even worse, it cannot identify the
 number of equilibria from the subgraph:
 the full network has one attractor, but the subgraph has two.
 It means that, with the wrong initial conditions,
 the biologist would conclude that the five species are going extinct,
 when in fact they are all doing fine in the real network.
 
  \emph{Our Contributions.}
  We give the first method to accurately learn
  steady-state dynamics when only a part of the network is observed.
  This is remarkable given the inherent interactive nature of the dynamics.
  The result in Figure~\ref{fig:example-5} demonstrates that our
  methods extract very close approximations to the true steady state dynamics
  \emph{in the full network} when just 5\% of vertices are observed.
  This surprising result has the potential
  for huge impact since up to now,
  the state-of-the-art is the naive approach which produces completely wrong conclusions.
  There are three main ideas behind our method.
  \begin{itemize}
  \item
    A mean field approximation to account for the impact of the unobserved part of the network.
  \item
    Summarizing the mean field impact using a 
    \emph{resilience} parameter, \math{\beta}, which 
    depends only on network topology. How the resilience
    impacts the final outcome depends on the coupled nonlinear
    dynamics through \math{f(\cdot)} and \math{g(\cdot,\cdot)}.
  \item
    Estimating the \emph{full network's} resilience from the observed subgraph.
    A network's resilience  is important in other contexts. The resilience characterizes a complex system's ability to
retain its basic functionality under edge and vertex faults. Hence, our estimates of resilience from incomplete information are of independent interest.
  \end{itemize}
Combining these three ideas, we obtain accurate
estimates of the steady-states as in Figure~\ref{fig:example-5}. Our
estimates are near-exact matches to the true steady-states.

\section{Model}

The true dynamics on the full network \math{G}
are governed by
the coupled nonlinear dynamics in
\r{eq:general-ODE}. We represent \math{G} by its adjacency
matrix \math{A}, and assume that the total size of the network,
\math{n}, is known. The observed sampled subgraph
\math{G^{(s)}=(V^{(s)},E^{(s)})} has adjacency matrix \math{A^{(s)}}.
There are many ways to sample vertices and edges from a graph.
We focus on
two natural sampling methods which are reasonable models of how the incomplete
data is often obtained.
\begin{itemize}
\item {\bf (Random Vertex Sampling)} Form the induced subgraph for
  randomly sampled vertices. We assume the degrees 
  of the sampled vertices are also known. For example, we know the number
  of friends each person has in a social network and who is friends with whom
  among the sampled subnetwork.
\item {\bf (Random Walk)} A random vertex is sampled. At each step,
  an available edge is followed to sample a new vertex.
  The degrees
  of the vertices are implicitly
  available since, at each step, the available edges must
  be known.
\end{itemize}
Our method can be extended to other sampling schemes, such as edge sampling, degree
biased sampling, Metropolis-Hastings random walks.
The main property we require of the sampling is 
that specific topological parameters of the graph can be reliably estimated
from the sample.

For \math{i\in V^{(s)}}, the
steady-state value in \math{G} is denoted \math{x_i^*}. We denote by
\math{z_{i}^*} the steady-state value obtained using the partially
observed network. We loosely use \math{z_i} to refer both to the vertex
and state variable at the vertex.
The naive method solves the same 
system in \r{eq:general-ODE} for \math{A^{(s)}} instead of
\math{A}. That is, for \math{i\in V^{(s)}}, \math{z_i^*}
is the steady-state of the dynamical system
\mld{
 \text{(Naive method)}~~~~~~\dot{z}_{i}
  =
  f(z_{i})+\sum_{j\in V^{(s)}}A^{(s)}_{ij}g(z_{i},z_{j}).
\label{eq:naive}}
This naive approach produces incorrect conclusions, yet
it is common practice because that's all practitioners have.
To get correct results, one \emph{must} account for missing data.

\section{Mean Field Approximation and Resilience}
The main idea is shown in the (sampled) subgraph
in Figure~\ref{fig:mean-field-approx}(a).
We focus
on one vertex in \math{V^{(s)}},  \math{z_2}.
Vertex \math{z_2} interacts with
its neighbors  \math{z_1} and \math{z_4} in the subgraph, and
its neighbors outside the subgraph. All the subgraph nodes
\math{z_1,\ldots,z_5} are in a similar situation. Now fix the value
for each external neighbor of \math{z_2} to
its \emph{true steady-state value in the full network}, shown as
\math{x_1^*,x_2^*,x_3^*}. Do the same for the external
neighbors of all the subgraph vertices \math{z_1,\ldots,z_5}.
As far as the subgraph is concerned, all external neighbors have
converged to their steady-state values and are providing the right
interactive feedback to all subgraph nodes. The subgraph is effectively
isolated from the rest of the network and will
converge to steady-states of the full  network.
We state the next theorem without proof.
The preceding discussion essentially
amounts to the proof.
\begin{figure}[h]
\centering
\includegraphics[width=\linewidth]{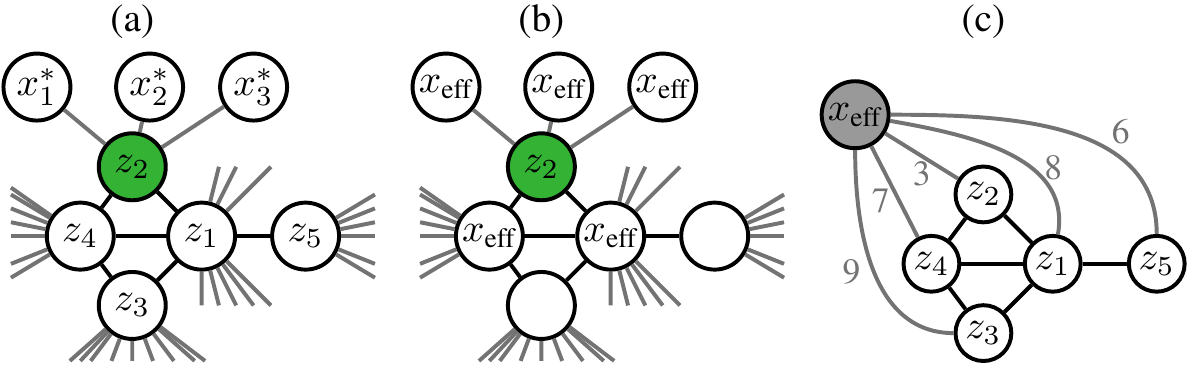}
\caption{Mean field approximation.\label{fig:mean-field-approx}}
\end{figure}

\begin{theorem} The steady-states \math{z_i^*}
  of the dynamical system
  \mld{
    \dot{z}_i=f(z_i)+\sum_{j\in V^{(s)}}A^{(s)}_{ij}g(z_i,z_j)
    +\sum_{j\not\in V^{(s)}}A_{ij}g(z_i,x_j^*)\label{eq:MF1}
    }
  recover one of the true steady-states of the full network.
\end{theorem}

We cannot implement~\eqref{eq:MF1}, because $x_i^*$ are unknown.
In the mean field approximation,
we replace \math{x_i^*} by an average influence
\math{x_\text{eff}}, see Figure~\ref{fig:mean-field-approx}(b).
Of course it is an approximation, but it works well for analyzing complex interacting systems
such as spin-systems~\cite{edwards1975theory}.
The iteration-0 estimate is  \math{x_\text{eff}} for all states.

Fix the states for all
vertices but (say) \math{z_2} to \math{x_\text{eff}},
and find the steady-state for \math{z_2}, as in Figure~\ref{fig:mean-field-approx}(b).
Now, \math{z_2} is effectively isolated from the rest of the
network as all its neighbors are \emph{fixed} at \math{x_\text{eff}}. 
So, \math{z_2} evolves to a steady-state, following the dynamics
\math{\dot{z}_2=f(z_2)+5g(z_2,x_\text{eff})}. Repeat for
each subgraph-vertex to arrive at
\math{z_i^{(1)}}, the steady-states of
\mld{
\label{eq:uncoupled}
  \dot{z}_i=f(z_i)+\delta_i g(z_i,x_\text{eff}),
}
where \math{\delta_i} is the degree of \math{z_i} in \math{G}.
These equations are uncoupled since \math{x_\text{eff}}
is \emph{fixed}. This is
the method used in \cite{gao2016universal}
to analyze the dynamics on the full network
by reducing to \math{n} uncoupled ODEs. We now iterate further.
Suppose the steady-state from iteration \math{\tau} is \math{z_i^{(\tau)}}.
We obtain \math{z_i^{(\tau+1)}}, the approximation at
iteration \math{\tau+1} as the steady-state solution to
the \emph{uncoupled} equations
\mld{
\label{eq:MF2}
\resizebox{.9\columnwidth}!{
$\dot{z}_i=f(z_i)+\sum_{j\in V^{(s)}}A^{(s)}_{ij}g(z_i,z_j^{(\tau)})+
  \sum_{j\not\in V^{(s)}}A_{ij} g(z_i,x_\text{eff})$.
}}
Comparing \r{eq:MF2} to the exact solution in \r{eq:MF1},
the external forces are replaced by an effective external
force and the interaction term is approximated by an interaction to
a previous steady-state. Iterating to convergence,
\math{z^{(\tau)}_i} converges to a steady-state \math{z_i^*} which
solves the coupled system
\mld{
\label{eq:MF3}
\resizebox{.9\columnwidth}!{
  $\dot{z}_i=f(z_i)+\sum_{j\in V^{(s)}}A^{(s)}_{ij}g(z_i,z_j)+
  (\delta_i-\delta_i^{(s)})g(z_i,x_\text{eff})$,
}}
where \math{\delta_i^{(s)}} is the degree of \math{z_i} in \math{G^{(s)}}.

The naive method in
\r{eq:naive} resembles \r{eq:MF3} with one crucial difference,
an additional term for the external force on a vertex.
The entire system in
\r{eq:MF3} corresponds to adding just one more vertex to our subgraph,
whose value is fixed at \math{x_\text{eff}},
with a weighted edge to
\math{z_i} of weight
\math{A^{(s)}_{i,x_\text{eff}}=\delta_i-\delta_i^{(s)}}. We show this
augmented graph for our example in Figure~\ref{fig:mean-field-approx}(c).
\myfbox[filllightgray]{0.425\textwidth}{\parbox{\linewidth}{\emph{\bf
    To account for missing vertices,
    add \emph{one} vertex to the subgraph, fixed 
   to \math{x_\text{eff}} and add degree-weighted edges from
    \math{x_\text{eff}} to all vertices in the subgraph.}}}
Next, we discuss how to compute \math{x_\text{eff}} to
estimate the mean-field
interaction with unseen vertices.
The complication is that this estimate cannot depend on the
missing information.
This is possible because \math{x_\text{eff}}
depends on the missing information
only through global topological statistics of the
network, and we can estimate those topological statistics when the subgraph
is sampled appropriately.

\subsection{Computing the Effective External Impact}
\label{sec:xeff}

There are two unknown quantities in \r{eq:MF3}. The degree \math{\delta_i} and
\math{x_\text{eff}}. We now discuss  \math{x_\text{eff}}, following the general approach
in~\cite{gao2016universal}. Consider vertex \math{i} and the interaction term
\math{\sum_{j} A_{ij} g(x_i,x_j)} in \r{eq:general-ODE}, where 
\math{A_{ij}} is the influence \math{j} has on \math{i}. The in-degree \math{\degin_i=\sum_{j}A_{ij}} and the out-degree
\math{\degout_i=\sum_{j}A_{ji}}. Assuming \math{A_{ij}\ge 0}, the interaction term is
the in-degree times an average interaction,
\mld{
  \sum_j A_{ij} g(x_i,x_j) = \degin_i\frac{\sum_{j} A_{ij} g(x_i,x_j)}{\sum_{k} A_{ik}}
}
Here, the in-degree \math{\degin_i}
captures the idiosyncratic part, and the average \math{g(\cdot,\cdot)} captures
the network effect.
Our first mean-field approximation is to replace
local averaging with global averaging, which approximates the
network-impact on a vertex as nearly homogeneous. Specifically, we have
\mld{
\label{eq:MF_av}
\resizebox{.9\columnwidth}!{
$\frac{\sum_{j} A_{ij} g(x_i,x_j)}{\sum_{k} A_{ik}}
\approx
\frac{\sum_{ij} A_{ij} g(x_i,x_j)}{\sum_{i k} A_{i k}}
=
\frac{\bm{1}^T A g(x_i,\bm{x})}{\bm{1}^T A \bm{1}}$,
}}
where the vector \math{g(x_i,\bm{x})} has 
\math{j}th component
\math{g(x_i,x_j)}. Define an averaging linear operator
\mld{\cl L_A({\bm z})=\frac{\bm{1}^T A }{\bm{1}^T A \bm{1}}{\bm z}
=\frac{\bm{s}^{\text{out}}\cdot {\bm z}}{\bm{s}^{\text{out}}\cdot \bm{1}},}
which is
a weighted average of the entries in \math{\bm z}.
Our  mean-field approximation results in the approximate dynamics
\mld{
\dot{x}_i=f(x_i)+\degin_i \cl L_A[g(x_i,\bm{x})].
\label{eq:approx1}}
In the first order linear
approximation, we take \math{\cl L_A} inside \math{g}.
Our second mean-field approximation is that the 
average of external interactions is approximately the interaction with
the average.  That is
\math{\cl L_A[g(x_i,\bm{x})]\approx g(x_i,\cl L_A(\bm{x}))} and
\mld{
\dot{x}_i=f(x_i)+\degin_i g(x_i,\cl L_A(\bm{x})),
\label{eq:approx2}}
where \math{\cl L_A(\bm x)} is a global state. Let \math{x_\text{av}\triangleq\cl L_A(\bm x)}.
Applying
\math{\cl L_A} to both sides of
\r{eq:approx2} gives
\mld{
\dot{x}_\text{av}=\cl L_A[f(\bm x)]+
\cl L_A[\bm\degin g(\bm x,x_{\text{av}})].}
Extensive tests in~\cite{gao2016universal} show that the
in-degrees \math{\degin_i} and the interactions \math{g(x_i,x_{\text{av}})} are roughly
uncorrelated, so the \math{\cl L_A}-average of the product is
the product of
\math{\cl L_A}-averages.
Thus, our third mean-field approximation is
\math{\cl L_A[\bm\degin g(\bm x,x_{\text{av}})]
\approx\cl L_A(\bm\degin)\cl L_A[g(\bm x,x_{\text{av}})]}.
Using the linear approximation again, we take the \math{\cl L_A}-average
inside \math{f} and \math{g}
\mld{
\dot{x}_{\text{av}}
=
f(\cl L_A(\bm x))+
\cl L_A(\bm \degin) g(\cl L_A(\bm x),x_{\text{av}}).
}
Now we have a dynamical system for \math{x_{\text{av}}},
\mld{
\dot{x}_{\text{av}}=
f(x_{\text{av}})+
\beta g(x_{\text{av}},x_{\text{av}}),\label{eq:dotxeff}}
where the resilience \math{\beta=\cl L_A(\bm{s}^{\rm in})}.
For undirected graphs,
\math{\beta=\sum_i\delta_i^2/\sum_i\delta_i
=\langle\delta^2\rangle/\langle\delta\rangle}.
The steady-state of \r{eq:dotxeff} is the external effective impact, \math{x_{\text{eff}}}.
Plugging it into
\r{eq:approx2} gives
an uncoupled ODE for \math{x_i},
\mld{
\dot{x}_i=f(x_i)+\degin_i g(x_i,x_{\text{eff}}).\label{eq:dotxi}}
In the mean-field
approximation, \math{g(x_i,x_j)} in 
\r{eq:general-ODE} is replaced by an interaction with a mean-field
external world
\math{g(x_i,x_{\text{eff}})} and the number of
neighbors impacting \math{x_i} is captured by  \math{\degin_i}.
To approximately obtain the steady-states of
the system, one first solves the ODE in
\r{eq:dotxeff} to get \math{x_{\text{eff}}},
and then \math{n} uncoupled ODEs at each vertex
to get \math{x_i}, which only depends on
\math{\degin_i} if given \math{x_{\text{eff}}}. The method works well because the mean-field approximations
only need to hold \emph{at the steady-state}. Hence, we can recover the steady-state for any vertex (for
example the sampled vertices) from accurate estimates
of degrees \math{\degin_i} (\math{\delta_i} in the undirected case) and the resilience parameter $\beta$.

\subsection{Evaluating the Mean-Field Approximation}
One non-trivial implication of our mean-field approach is that the steady-states
are approximated by solving the uncoupled equations in
\r{eq:uncoupled}. The parameter \math{x_{\text{eff}}} in \r{eq:uncoupled}
only depends on
\math{\beta=\langle \delta^2\rangle/\langle \delta\rangle}.
So the ODE in \r{eq:uncoupled} only depends on 
the degree sequence of the original  network, which means the steady-states can be approximated by knowing only a network's degree sequence. We verify this in
Figure~\ref{fig:rewire-heter}(a), which compares
true steady-states in the ecological network with
steady-states of a random network that preserves the
degree sequence. The near-perfect matching of the steady-states  
is empirical evidence for our  mean-field approach.

\subsection{Accuracy of Our Approach}

The mean-field approximation essentially replaces individual interactions with
an average, and amounts to a homogeneity assumption.
The more homogeneous a network, the more accurate our approximations. Indeed, the
method of
solving for \math{x_{\rm{eff}}} as a steady state of
\r{eq:dotxeff} and then for \math{x_i^*} as steady states of
\r{eq:dotxi} produces an exact solution
for a regular network (perfectly homogeneous).
The next theorem says that our method in~\r{eq:MF3} is 
perfect for such regular networks.
\begin{theorem}
For a \math{k}-regular network,
the steady-states $z_i^*$ obtained by solving the dynamical system in~\r{eq:MF3} with
\math{x_{\rm{eff}}} obtained as a steady-state of \r{eq:dotxeff} with \math{\beta=k} 
recovers an exact steady-state $x_i^*$.
\end{theorem}
\begin{proof}(Sketch)
\math{x_{\text{eff}}} is a steady state of \r{eq:dotxeff}
with \math{\beta=k}, hence  \math{f(x_{\rm{eff}})+
k g(x_{\rm{eff}},x_{\rm{eff}})=0} as  \math{\dot{x}=0}.
We show that \math{x_i^*=x_{\rm{eff}}} for \math{i\in[1,N]}
is a fixed point of the system.
Since node $i$ has $k$ neighbors and each of them has state \math{x_j^*=x_{\rm{eff}}},
\math{\dot{x}_i=f(x_i)+ \beta g(x_i,x_{\rm{eff}})=0} when \math{x_i=x_{\rm{eff}}}.  
Lastly, \r{eq:MF3} converges to \r{eq:dotxeff},
because \math{x_i=x_{\rm{eff}}}.
\end{proof}
Essentially, our approach is perfect for regular
networks. The degree
of inhomogeneity in the network is therefore a parameter which controls the
quality of approximation. We now show some experimental results with synthetic networks
where we can control for the inhomogeneity. Even for extremely inhomogeneous
networks, our approach suffers little, indicating the strength of the mean-field
approach.

\begin{figure}[ht]
\centering
\begin{tikzpicture}
\node[scale=1](g)at(0,0){
\includegraphics[width=0.43\textwidth]{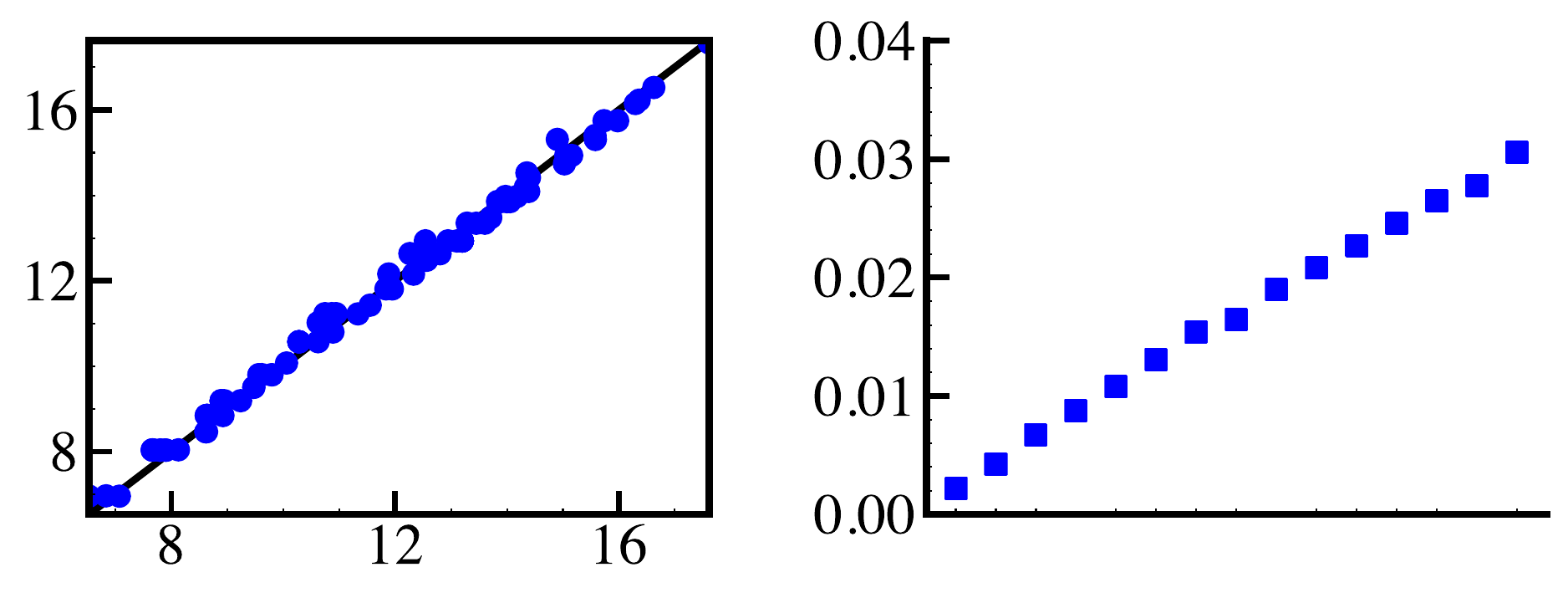}};
\node at(-2,-1.6){True};
\node[rotate=90] at(-3.9,0){Rewired};
\node[scale=0.75] at (1.2,-1.3){Erd\H{o}s-R\'enyi};
\node[scale=0.75] at (3.4,-1.3){scale-free};
\node[rotate=90] at(-0.1,0){Relative error};
\node at(2.2,-1.6){Heterogeneity};
\node at(-2,1.5){(a)};
\node at(2,1.5){(b)};
\end{tikzpicture}
\caption{(a) \math{G^{(\text{rewired})}} vs. \math{G}; (b) Impact of heterogeneity.\label{fig:rewire-heter}}
\end{figure}
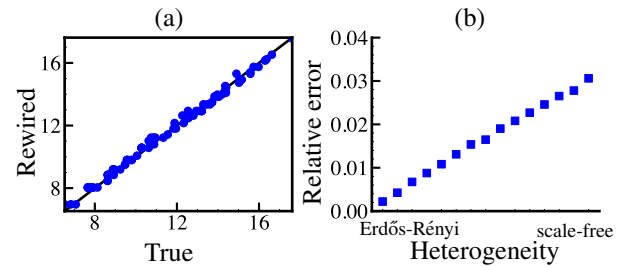
To evaluate the impact of heterogeneity,
we use 15 random 1000-vertex networks
with
different heterogeneities \math{\cl H} measured by
the relative degree dispersion,
\math{\cl H\triangleq [\langle \delta^2\rangle-\langle \delta\rangle^2]/\langle \delta\rangle}.
Relative state estimation errors for
our method in~\r{eq:MF3} for an observed subgraph of just 10 vertices and
estimates of \math{\beta} and \math{x_{\text{eff}}} from
Section~\ref{section:estimate-resilience}
are shown in Figure~\ref{fig:rewire-heter}(b).
As expected, performance deteriorates as heterogeneity increases.
The far right is
the very heterogeneous scale-free network and leftmost is the
nearly homogeneous Erd\H{o}s-R\'enyi
network. Even for very heterogeneous networks, the relative error is only about 3\%.

\section{Estimating the Resilience}\label{section:estimate-resilience}
To get \math{x_\text{eff}}, we solve for the steady-states of \r{eq:dotxeff},
so we need to estimate the resilience
\math{\beta={\langle \delta^2\rangle}/{\langle \delta\rangle}},
a topological statistic of the full network.
Resilience is well studied
in science and engineering,
arising in many contexts because it
captures a complex system's ability to
retain its basic functionality under faults~\cite{gao2016universal}.
Understanding a network's resilience is essential for us to evade the 
consequences of resilience loss, such as malfunction of gene regulation 
networks, cascading failures in technological systems, mass extinctions 
in ecological networks.

We estimate resilience \math{\beta} from the sampled subgraph \math{G^{(s)}}.
In other contexts where it is important to measure resilience,
the full network topology is also often not available
perhaps 
due to privacy, or, as is usually
the case in practice, the inability to measure the full network
(for example, protein-protein interactions,
metabolic and terrorist networks~\cite{stumpf2005sampling}).
Despite advances in graph sampling,
  we are not aware of accurate estimators of resilience from
  an incomplete view of the network.
  The naive resilience-estimator treats the observed subnetwork as
  the full network and estimates \math{\beta} with \math{\beta^{(s)}}.
  We derive corrections to this naive estimate for a variety of
  sampling schemes, and give analysis of the estimation accuracy (bias and variance)
  and the sample complexity.
  This general infrastructure enables one to manage a network's resilience
from incomplete data, and may be of independent interest.

We treat resilience estimation for
undirected, unweighted graphs. However, our results can be
easily extended to the directed and weighted cases.
This part of our work falls into the general area of 
graph analysis from incomplete (sampled) subgraphs.
Typical sampling methods are
vertex-based, exploration-based
(see for example
\cite{leskovec2006sampling,hubler2008metropolis,ahmed2011network})
and edge-based. We focus on the first two sampling schemes
since they are natural ways to
sample a graph in practice, however, the results do extend to
edge-based sampling.
The birds-eye view of the
workflow is

\begin{center}
\begin{tikzpicture}[>=latex,line width=1pt,x=1.2cm]
\node[draw,rounded corners=5pt,inner sep=2pt,scale=0.9](G)at(0,0){$G,\beta$};
\node[draw,rounded corners=5pt,inner sep=2pt,anchor=west,scale=0.9](Gs)at($(G.east)+(1.4,0)$){$G^{(s)},\beta^{(s)}$};
    \node[inner sep=3pt,anchor=west,scale=0.9](B)at($(Gs.east)+(1.2,0)$){$\hat\beta(G^{(s)},\beta^{(s)},\bm{\Phi})\approx\beta$,};
    \path
    (G)edge[->]node[above,scale=0.75,inner sep=1pt]{sampler $\bm{\Phi}$}(Gs)
    (Gs)edge[->]node[above,scale=0.75,inner sep=1pt]{this work}(B)    
    ;
\end{tikzpicture}
\end{center}
where $\beta^{(s)}$ is  the naive resilience estimator that treats 
$G^{(s)}$ as if it were the complete network.
Our final estimator $\hat\beta$ can depend on the sampling method
$\bm{\Phi}$.

\subsection{Estimating $\beta$ for Random Subgraphs}

We consider several types of subgraph sampling. The simplest
is to sample $m$ vertices uniformly without replacement and
measure the
degree $\delta_i$ of each vertex. The sample averages
$\langle \delta^2\rangle_s$ and \math{\langle \delta\rangle_s} are unbiased
and concentrate at the
true averages \math{\langle \delta^2\rangle} and \math{\langle \delta\rangle}.
This means \math{\hat\beta=\langle \delta^2\rangle_s/\langle \delta\rangle_s}
is asymptotically unbiased and concentrates at \math{\beta}.
We use \math{\langle \cdot\rangle_s} to denote the average
over the subgraph vertices in \math{V^{(s)}}.

An interesting variant is to sample vertices with a degree-bias,
so the probability to sample vertex \math{i} is
\math{\delta_i/\sum_i\delta_i}. In this case, the expected degree of a sampled
node is \math{\sum_i\delta_i^2/\sum_i\delta_i=\beta}, hence \math{\hat\beta=\langle \delta\rangle_s} is an unbiased estimate of \math{\beta}.
More generally, one can sample vertices with an arbitrary degree-biased
importance distribution \math{q(\delta)}. The analysis
of this more general case, including a bias-variance analysis is postponed to a full
version of this work.

Now, suppose you only
measure \math{\delta^{(s)}_i}, the induced degrees in the subgraph,
not true degrees. To implement
\r{eq:MF3} we need not only \math{\beta}, but also the degree \math{\delta_i}.
Let  \math{\beta^{(s)}} be the resilience of the induced subgraph,
\math{\beta^{(s)}
  =\langle \delta^2\rangle_s/\langle\delta\rangle_s.
  }
A crude analysis just uses concentration twice.
First, the induced degree \math{\delta_i^{(s)}}
is a sum of \math{m-1} Bernoulli random variables
sampled uniformly without replacement from
a population of \math{n-1} Bernoulli values in which \math{\delta_i} of them
are 1. So, \math{\delta_i^{(s)}/(m-1)} concentrates at \math{\delta_i/(n-1)}.
We need concentration for each of the \math{m} sampled vertices,
which, using a union bound plus a Hoeffding inequality, has a
failure probability \math{2m\exp(-\Omega(m\varepsilon^2))} for relative
error \math{\varepsilon}.
So, the estimates
\math{\hat\delta_i=\frac{n-1}{m-1}\delta_i^{(s)}} all concentrate with
relative error at their respective
\math{\delta_i} and hence we get a resilience estimate
\math{\hat\beta=\langle\hat\delta^2\rangle/\langle\hat\delta\rangle=
\frac{n-1}{m-1}\beta^{(s)}} that 
concentrates at \math{\beta}.
A more refined estimate based on
\math{\mathbb{E}[(\delta_i^{(s)})^2]} is
\math{\hat\beta=\frac{n-2}{m-2}\beta^{(s)}-\frac{n-m}{n-2}}.

Another popular way to sample a subgraph is using a random walk: 
start at a random vertex and 
move from one vertex to a neighbor, chosen uniformly at random
(nodes could be revisited).
Such a walk has a
stationary distribution where a node's sampling probability
is proportional to its
degree~\cite{kurant2011towards}.
From the discussion
of degree-biased sampling, the estimator of resilience is
\math{\hat\beta=\langle \delta\rangle_s}
when the degrees in the full network are known.
When only the induced subgraph is known, then, as with the induced subgraph
from vertex sampling, a correction factor is needed.
We approximate the sampling as independent degree-biased sampling in a random
rewiring model, and get 
a correction factor \math{n/m} (see full version).
We summarize our
discussion of random vertex sampling (VS) and random walks (RW)
in a table.
\begin{table}[ht]{\small\tabcolsep6pt\renewcommand{\arraystretch}{1.2}
\begin{tabular}{r|lll}
       Sampling& Measured& \math{{\hat\delta}_i}&
        \math{\hat{\beta}}\\\hline
      VS&\math{V^{(s)}, E^{(s)}, \delta_i}&\math{\delta_i}&\math{\langle \delta^2\rangle_s/\langle \delta\rangle_s}\\
       Ind-VS&\math{V^{(s)}, E^{(s)}, \delta_i^{(s)}}&\math{\textstyle\frac{n-1}{m-1}\delta_i^{(s)}}&\math{\textstyle\frac{n-2}{m-2}\beta^{(s)}-\frac{n-m}{n-2}}\\
      RW&\math{V^{(s)}, E^{(s)}, \delta_i}&\math{\delta_i}&\math{\textstyle\langle \delta\rangle_s}\\
      Ind-RW&\math{V^{(s)}, E^{(s)}, \delta_i^{(s)}}&\math{\textstyle\frac{n\langle\delta\rangle}{m\hat\beta} \delta^{(s)}_i}&\math{\textstyle\frac{n}{m}\langle \delta\rangle_s}
      \\\hline
      \multicolumn{4}{r}{\scalebox{0.775}{\parbox{0.55\textwidth}{
      For \math{i\in V^{(s)}}, \math{\delta_i} is the degree in 
\math{G}, \math{\delta^{(s)}_i} is the degree in \math{G^{(s)}} and the
resilience of the
subgraph is
\math{\scriptstyle\beta^{(s)}=\langle\delta^2\rangle_s/\langle\delta\rangle_s}.
The notation \math{\langle x\rangle_s} means \math{\scriptstyle\sum_{i\in V^{(s)}}x_i/m}.
}}}
\end{tabular}}
\label{eq:tab:estimates}
\end{table}

\begin{table*}[ht]
\begin{tabular}{r|l}
Applications  & Networks \math{(|V|,|E|)} \\\hline
Ecological & \textbf{ENet1}(270,8074) \cite{arroyo1985community,gao2016universal} \\
&\textbf{ENet2}(97,972) \cite{clements1923experimental,gao2016universal}\\\hline
Gene Regulation & \textbf{MEC}(2268,5620) \cite{lee2002transcriptional,gao2014target}\\
&\textbf{TYA}(662,1062) \cite{lee2002transcriptional,gao2014target} \\\hline
Epidemic & \textbf{Dublin}(410,2765) \cite{rossi2015network} \\
& \textbf{Email}(1133,5451) \cite{guimera2003self,kunegis2013konect}\\\hline
\end{tabular}
\caption{List of networks in our evaluation}
\label{tab:networks}
\end{table*}

\begin{figure*}[ht]
\centering
\includegraphics[width=.95\textwidth]{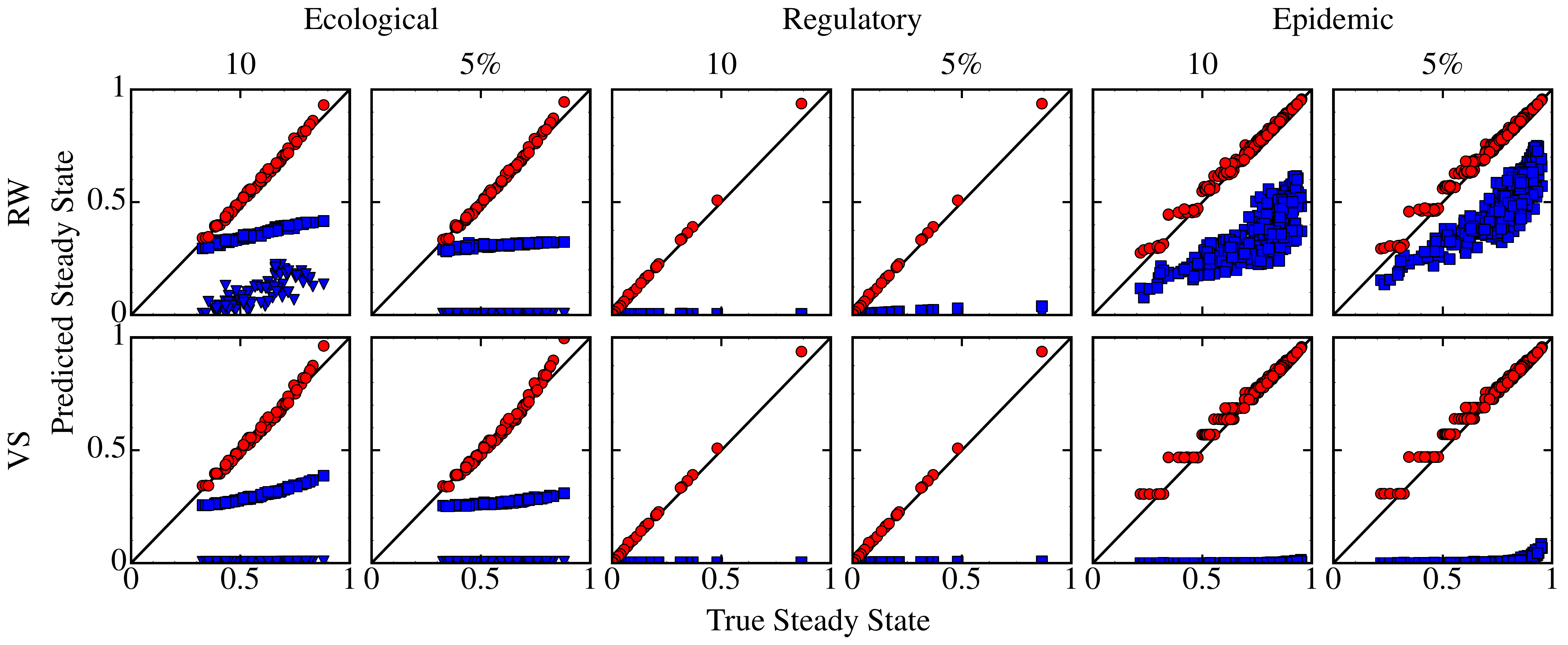}
\caption{Predicted steady states when vertex degrees
  \math{\delta_i} are available. The diagonal line is perfect prediction.
  Our method (red) is nearly perfect, while the naive method (blue) is a
  disaster. Often, the steady states just converge to 0 for
  the naive method.}
\label{fig:predvsreal}
\end{figure*}

\begin{figure*}[ht]
\centering
\includegraphics[width=.9\textwidth]{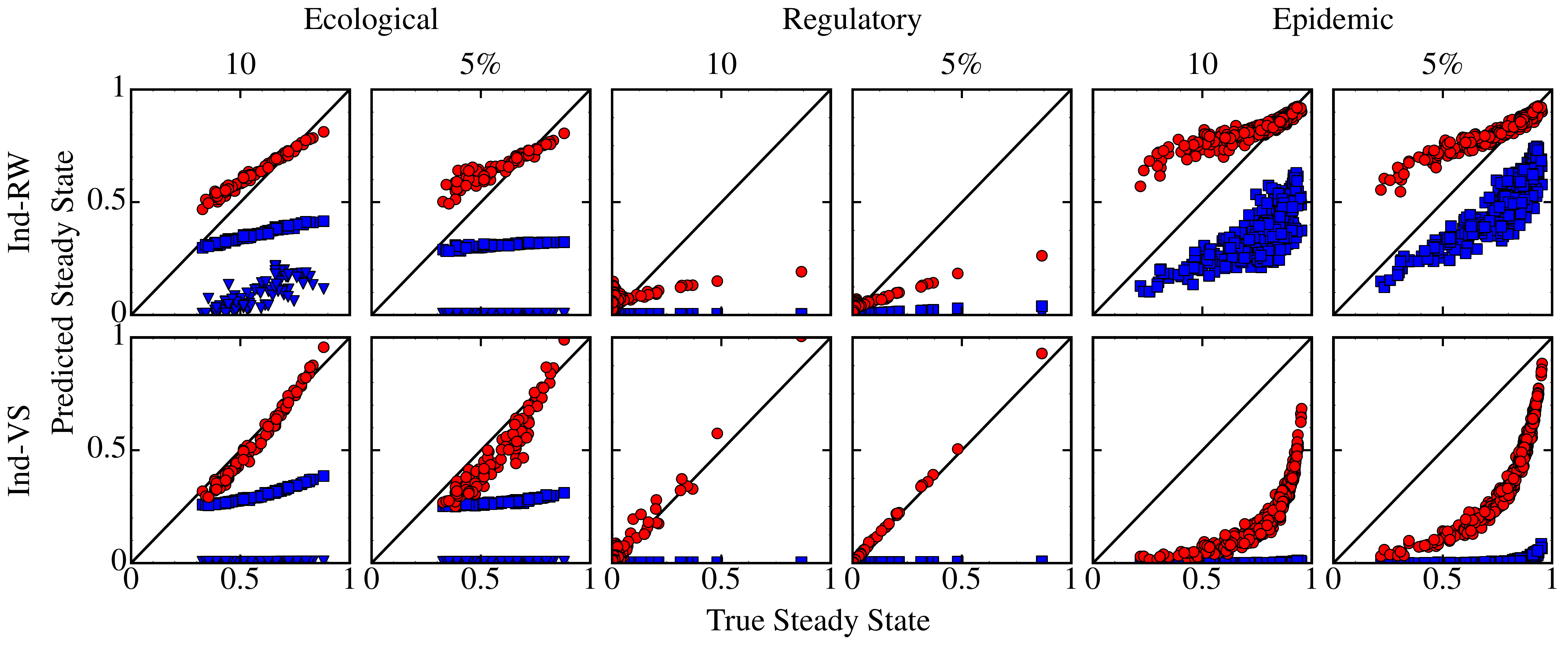}
\caption{Predicted steady-states when only the induced subgraph is
    available and vertex degrees must be estimated.
    Our method (red)
    significantly outperforms naive, but  performance
    drops compared to Figure~\ref{fig:predvsreal}, because
    estimating vertex degrees from observed
    induced degrees is tough.
\label{fig:indpredvsreal}}
\end{figure*}

For Ind-RW, the network average degree
\math{\langle\delta\rangle}
is
needed  to estimate true degrees (in addition to \math{n}).
There are methods to estimate \math{\langle\delta\rangle}
\cite{Dasgupta2014,ZKS2015,RT2010,leskovec2006sampling},
but they use more powerful queries than the induced
subgraph from a simple random walk can provide.
The complication with the simple random walk is that it is a
degree-biased sampling of vertices. This bias
can be corrected with
a Metropolis-Hastings random walk~\cite{hubler2008metropolis}, but
that requires knowledge of
vertex degrees.

One can analyze the bias and variance of
our estimators using the importance sampling
framework that samples vertices according to the proposal distribution
\math{q}~\cite{theodoridis2015machine,cochran1977sampling,wu1982estimation} (uniform or degree-biased).
We postpone the details to the full paper.


\section{Results}
We tested our approach on the three popular dynamical systems in
Table~\ref{tab:real-dynamics} and two corresponding
networks for each dynamical system (see Table~\ref{tab:networks}).

Each dynamical system contains several parameters which are
set as in~\cite{gao2016universal} for ecology and gene regulation,
and as in~\cite{barzel2013universality} for epidemics.
We compare the performance of different subgraph-sampling
(RW, VS, Ind-RW, Ind-VS), with sample size \math{m=10} (constant size)
and \math{m=0.05n} (constant fraction).
We average results
over several subgraphs. When a vertex is sampled in multiple subgraphs,
we report the average predicted steady-state.

Given a subgraph,
we solve \r{eq:MF3} to estimate the true steady-states of vertices
in the subgraph, for which we need \math{\delta_i} (true degree)
and \math{x_{\text{eff}}} (effective external state).
To get~\math{x_{\text{eff}}}, we need \math{\beta}
to solve~\r{eq:dotxeff}.
For \math{\delta_i} and \math{\beta}, we
use the estimates in Section \ref{eq:tab:estimates}.
In solving ~\r{eq:dotxeff}, there can be one stable attractor or
more than one stable attractors (there
can also be unstable attractors). When there are more than one
stable attractor, the system can equilibrate at multiple steady-states.

First, we consider sampling methods that obtain
vertex degrees \math{\delta_i} (RW, VS), where
only \math{\beta} needs to be estimated. The results in
Figure~\ref{fig:predvsreal} show that our
approach (red) is remarkable at revealing the true steady-state, even
from tiny subgraphs. We get near perfect results from just
10 vertices of a multi-thousand vertex graph. 
This is all the more impressive given that the current
state-of-the-art, i.e., the naive approach (blue),
is more or less a disaster.
We emphasize that the naive method is the only method currently used by researchers
in the field, and it simply fails.
In the ecological application, the naive
method even identifies multiple steady-states, one of which
sends all species to extinction.

When only induced subgraphs are available (Ind-RW, Ind-VS),
see  Figure~\ref{fig:indpredvsreal},
our methods are still much better than naive, but the performance drops.
Estimating individual degrees from induced degrees is hard.
Parameters like the resilience \math{\beta} are global and
can be extracted accurately from subgraphs. Local parameters
like vertex degree get severely distorted. More details on
how degree and resilience estimations are affected by induced subgraph sampling
is deferred.
For vertex sampling, our estimator is unbiased, but for random walks,
our estimators are based on the approximation of independent degree-biased
sampling. This approximation can break down.
\mand{
  \parbox{0.425\textwidth}{
    \emph{Open Question.} How should one estimate vertex degrees from
    induced subgraphs of random walks? How
    does the estimator depend on the network properties?}}
    We also observe from Figure~\ref{fig:indpredvsreal} that for induced
    subgraphs, the performance
    of our approach (and the naive method) depends strongly on the subgraph
    sampling
    methods, depending on the network structure
    and the nonlinear dynamics.

\mand{
  \parbox{0.425\textwidth}{
    \emph{Open Question.} What are the factors which
    influence the choice of subgraph sampling method?
    For example when is RW better than VS?
}}    

\section{Discussion}

We addressed a prevalent problem.
Consider this scenario.
A biologist has the favorite part of
an ecosystem, their favorite 10-species, and carefully collects their
relationships which are summarized in the adjacency matrix
\math{A^{(s)}}. The biologist even knows how species interact,
the dynamical system \r{eq:general-ODE}. The biologist carefully
simulates the system to steady-state and finds that all species are going
extinct. This is scary, but the result is just plain wrong. You cannot
restrict a coupled nonlinear system to your favorite part of the
network and expect even close to correct conclusions by just analyzing that
part in isolation. One solution is to collect the full network
and analyze the full system. There are two problems. First,
we can't collect the
full network. Second, simulating the full system to equilibrium is prohibitive
in terms of convergence time. So the only feasible solution
for learning the true steady states of the observed incomplete network 
is to somehow account
for the external impact on the local system. This was our approach. In
a mean field approximation, the external impact reduces to a single
parameter 
\math{x_{\text{eff}}} which depends only on the network's resilience $\beta$,
a topological parameter. We showed how to estimate resilience,
depending on
how the subgraph is sampled. Our results on real networks with
corresponding dynamics gave spectacular success -- we accurately recover
steady-states from just 10-vertex subgraphs of thousand-vertex networks.

There are several interesting future questions.
The natural one is to find
improved estimates of resilience that extend to
other sampling methods (snowball sampling, edge sampling etc).
A critical direction, which we address in forthcoming work, is the inverse
problem. Suppose the steady-states are known. For example,
the abundance of your 10 favorite monkey species is known.
Can one infer the correct
dynamical system \math{f(\cdot),g(\cdot,\cdot)}?
Currently, the dynamics are fit to the observed steady-states for the partial
network~\cite{ghahramani1999learning,schmidt2009distilling,ionides2006inference,bongard2007automated,brunton2016discovering}.
This is wrong and will produce the incorrect
dynamics.
It is no surprise, therefore, that the inferred
dynamics keeps changing as more data is collected
\cite{schmidt2009distilling,ionides2006inference,bongard2007automated,brunton2016discovering}.
It is absolutely
critical to account for the external impact in the inference
process.

\section{Acknowledgements}
This work was partially supported by the ONR Contract N00014-15-1-2640.

\fontsize{9.0pt}{10.0pt}
\bibliographystyle{aaai}
\bibliography{AAAI-JiangC.9161}

\end{document}